\documentclass[sigconf]{acmart}
\usepackage{balance}
\setcopyright{none}
\renewcommand\footnotetextcopyrightpermission[1]{}
\def\BibTeX{{\rm B\kern-.05em{\sc i\kern-.025em b}\kern-.08emT\kern-.1667em\lower.7ex\hbox{E}\kern-.125emX}}

\usepackage{amsmath,amssymb,xspace,verbatim,paralist,enumitem}
\usepackage{array,multirow,multicol}
\usepackage{algorithm,algorithmicx}
\usepackage[noend]{algpseudocode}
\usepackage{graphicx}
\usepackage{mathdots,mathtools,mathrsfs}
\usepackage{amsthm}

\usepackage{thmtools}
\usepackage{thm-restate}

\usepackage[nameinlink]{cleveref}
\usepackage{subcaption}

\newtheorem{theorem}{Theorem}[section]
\newtheorem{lemma}[theorem]{Lemma}
\newtheorem{conjecture}[theorem]{Conjecture}
\newtheorem{corollary}[theorem]{Corollary}

\theoremstyle{definition}  \newtheorem{definition}[theorem]{Definition}

\newbox\mybox 
\newdimen\myboxwidth    

\newcommand\addpicture[3]{% 
\setbox\mybox=\hbox{\includegraphics[scale=#3]{#2}}
\myboxwidth\wd\mybox    
\renewcommand\windowpagestuff{% 
\includegraphics[scale=#3]{#2}
\captionof{figure}{A test figure.}}
\parpic[#1]{% 
\begin{minipage}{\myboxwidth}
 \windowpagestuff 
\end{minipage} 
} }

\newcommand{\ignore}[1]{}

\newcommand\Tstrut {\rule{0pt}{3ex}}         % = `top' strut
\newcommand\Bstrut {\rule[-1.3ex]{0pt}{0pt}}   % = `bottom' strut

\newcommand{\mypar}[1]{\medskip\noindent{\sffamily\bfseries #1.}~}
\newcommand{\etal}{{et al.}\xspace}

\newcommand{\R}{\mathbb R}

\newcommand{\eps}{\varepsilon}

\newcommand{\poly}{\mathrm{poly}}

\newcommand{\EX}{\hbox{\bf E}}
\newcommand{\Var}{\hbox{\bf Var}}

\newcommand{\EE}{\mathbb{E}}

\newcommand{\DISJ}{\textsc{disj}\xspace}

\newcommand{\Tau}{\mathcal{T}}

\DeclareMathOperator{\heavy}{heavy}

\newcommand{\nt}{T}
\newcommand{\degen}{\kappa}

\newcommand{\tO}{\widetilde{O}}

\newcommand{\temax}{J}			% max num triangles containing an edge
\newcommand{\dmax}{\Delta}		% max degree in a graph
\newcommand{\vertexcover}{C}
\newcommand{\pathcount}{P_2}

\newcommand{\assign}{\tau}
\newcommand{\total}{\Tau}

\DeclareMathOperator{\argmin}{argmin}
\DeclareMathOperator{\fixed}{fixed}

\newcommand{\otilde}{\widetilde{O}}
\newcommand{\disj}{\textsc{disj}\xspace}

\begin{document}

\fancyhead{}
\title{How the Degeneracy Helps for Triangle Counting in Graph Streams}

\author{Suman K. Bera}
\email{sbera@ucsc.edu}
\affiliation{%
  \institution{UC Santa Cruz}
  \streetaddress{1156 High St}
  \city{Santa Cruz}
  \state{California}
  \postcode{95064}
}
\author{C. Seshadhri}
\email{sesh@ucsc.edu}
\affiliation{%
  \institution{UC Santa Cruz}
  \streetaddress{1156 High St}
  \city{Santa Cruz}
  \state{California}
  \postcode{95064}
}


\begin{abstract}
    We revisit the well-studied problem of triangle count estimation in graph streams. Given a graph represented as a stream of $m$ edges, our aim is to compute a $(1\pm\eps)$-approximation to the triangle count $T$, using a small space algorithm. For arbitrary order and a constant number of passes, the space complexity is known to be essentially $\Theta(\min(m^{3/2}/T, m/\sqrt{T}))$ (McGregor et al., PODS 2016, Bera et al., STACS 2017).
    
    We give a (constant pass, arbitrary order) streaming algorithm that can circumvent this lower bound for \emph{low degeneracy graphs}. The degeneracy, $\degen$, is a nuanced measure of density, and the class of constant degeneracy graphs is immensely rich (containing planar graphs, minor-closed families, and preferential attachment graphs). We design a streaming algorithm with space complexity $\widetilde{O}(m\degen/T)$. For constant degeneracy graphs, this bound is $\widetilde{O}(m/T)$, which is significantly smaller than both $m^{3/2}/T$ and $m/\sqrt{T}$. We complement our algorithmic result with a nearly matching lower bound of $\Omega(m\degen/T)$.

    % We give a constant pass algorithm that needs at most $O(m/T)$ space for bounded degeneracy graph, where $m$ in the number of edges and $T$ is a promised lower bound on the number of triangles. Our results improve the space complexity by a factor of a $O(\sqrt{m})$ --- for general graphs the multi-pass complexity of this problem in $\Theta(m^{1.5}/T)$ (McGregor et al., PODS 2016, Bera et al., STACS 2017). We also prove matching lower bounds for the class of bounded degeneracy graph.
\end{abstract}
\begin{CCSXML}
<ccs2012>
<concept>
<concept_id>10003752.10003809.10010055</concept_id>
<concept_desc>Theory of computation~Streaming, sublinear and near linear time algorithms</concept_desc>
<concept_significance>500</concept_significance>
</concept>
<concept>
<concept_id>10003752.10003809.10003635</concept_id>
<concept_desc>Theory of computation~Graph algorithms analysis</concept_desc>
<concept_significance>100</concept_significance>
</concept>
</ccs2012>
\end{CCSXML}

\ccsdesc[500]{Theory of computation~Streaming, sublinear and near linear time algorithms}
\ccsdesc[100]{Theory of computation~Graph algorithms analysis}
\keywords{Triangle counting, Streaming Model, Degeneracy}


\maketitle

\section{Introduction} \label{sec:intro}

Triangle counting is a fundamental algorithmic problem for graph streams. 
Indeed, the literature on this one problem is so rich, that its study is almost
a subfield in of itself. Since the introduction of this problem by Bar-Yossef et al~\cite{BarYossefKS02}, there has been two decades of research on streaming algorithms
for triangle counting~\cite{BarYossefKS02,Jowhari2005,Buriol2006,tsourakakis2009doulion,Manjunath2011,tsourakakis2011triangle,Kane2012,Pavan2013,Pagh2012,Braverman2013,garcia2014,cormode2014,McGregor2016,bera2017,kallaugher2017hybrid}. The significance of triangle counting
is underscored by the wide variety of fields where it is studied: database theory,
theoretical computer science, and data mining. 
From a practical standpoint,
triangle counting is a core analysis task in network science. Given the scale
of real-world graphs, this task is considered to be computationally intensive. In database systems, triangle counting is used for query size estimation in database join
problems (see~\cite{atserias2008size,assadi2018simple} for details).
These have led to the theoretical and practical study of triangle counting in a variety of computational models: distributed shared-memory, MapReduce, and streaming~\cite{ChNi85,ScWa05,ScWa05-2,tsourakakis2008fast,avron2010counting,kolountzakis2012efficient,chu2011triangle,SuVa11,arifuzzaman2013patric,SePiKo13,TaPaTi13}.

Despite the plethora of previous work in the streaming setting, the following
question has not received much attention. \emph{Are there ``natural" graph classes
that admit more efficient streaming algorithms for triangle counting?} This question
has a compelling practical motivation. It is well known from network science
that massive real-world graphs exhibit special properties. Could graph classes that 
contain such real-world graphs have ``better than worst-case" streaming triangle algorithms?

Motivated by these considerations, we study the problem of streaming triangle counting,
parametrized by the \emph{graph degeneracy} (also called the maximum core number). We defer the formal definition for later, but for now, it suffices to think of degeneracy as a nuanced measure of graph sparsity. The class of constant degeneracy
graphs is extremely rich: it contains all planar graphs, all minor-closed families of graphs, and preferential attachment graphs. The degeneracy of real-world graphs is well studied, under the concept of \emph{core decompositions}. It is widely observed that the degeneracy of real-world graphs is quite small, many orders of magnitude smaller than worst-case upper bounds~\cite{goel2006bounded,JS17,shin2018patterns,DBS18}.

In computational models other than streaming, the degeneracy is known to be relevant for triangle counting. From the perspective of running time of exact sequential algorithms,
a seminal combinatorial algorithm of Chiba-Nishizeki gives an $O(m\degen)$ time algorithm for exact triangle counting ($m$ is the number of edges, and $\degen$ is the degeneracy)~\cite{ChNi85}. Thus, for (say) constant degeneracy, this algorithm beats the best known running time bounds of more sophisticated matrix multiplication based algorithm (of course, the latter work for all graphs)~\cite{alon1997finding}. In distributed and query-based computational models, a number
of results have shown that low degeneracy is helpful in bounding communication or query complexities~\cite{SuVa11,FFF15,JS17,ERS20}.
This inspires the main question addressed by this paper.

\medskip

    \emph{Do there exist streaming algorithms for approximate triangle counting on low degeneracy graphs that can beat known worst-case lower bounds?}

\subsection{Our results and significance} \label{sec:result}

We focus on constant pass streaming algorithms, with arbitrary order. Thus, we think of the input graph $G = (V,E)$ represented as an arbitrary list of (unrepeated) edges. Our algorithm is allowed to make a constant number of passes over this list, but has limited storage. As is standard, we use $n$ for the number of vertices, $m$ for the number of edges, and $T$ for the number of triangles in $G$.

We first define the graph degeneracy. 
\begin{definition} \label{def:degen} The degeneracy of a graph $G$,
denoted $\degen(G)$, is defined as $\max_{\textrm{$G'$ subgraph of $G$}}\{\textrm{min degree of $G'$}\}$. In words, it is the largest possible minimum degree of a subgraph of $G$.
\end{definition}

In a low degeneracy graph, all induced subgraphs have low degree vertices.
The following procedure that computes the degeneracy is helpful for intuition. Suppose one iteratively removed the minimum degree vertex from $G$ (updating degrees after every removal). For any vertex $v$, consider the ``observed" degree \emph{at the time of removal}. One can prove that the degeneracy is the largest such degree~\cite{degen-wiki}. Thus, even though $G$ could have a large maximum degree, the degeneracy can be small if high degree vertices are typically connected to low degree vertices. 

Our main theorem follows.

\begin{theorem} 
\label{thm:main} Consider a graph $G$ of degeneracy at most $\degen$,
that is input as an arbitrary edge stream. There is a streaming algorithm that 
outputs a $(1\pm\eps)$-approximation to $T$, with high probability\footnote{We use ``high probability" to denote errors less than $1/3$.}, and has the following properties.
It makes constant number of passes over the input stream and uses space $(m\degen/T)\cdot \poly(\log n, \eps^{-1})$.
\end{theorem}

To understand the significance of the bound $m\degen/T$, note that space
complexity of streaming triangle counting is known to be $\min(m^{3/2}/T, m/\sqrt{T})$~\cite{bera2017,McGregor2016}. Consider $\degen = O(1)$, which as mentioned earlier, holds for 
all graphs in minor-closed families and preferential attachment graphs. In this case,
the algorithm of Theorem~\ref{thm:main} uses space $\otilde(m/T)$. This is significantly smaller than both $m^{3/2}/T$ and $m/\sqrt{T}$. We note that for all graphs, $\degen \leq \sqrt{2m}$, and thus, the space is always $\widetilde{O}(m^{3/2}/T)$.

As an illustrative example, consider the wheel graph with $n$ vertices (take a cycle with $n-1$ vertices, and add a central vertex connected to all other vertices). Note that $m = T = \Theta(n)$ and $\kappa = O(1)$ ($G$ is planar). The space bound given in Theorem~\ref{thm:main} is only polylogarithmic, while \emph{all existing} streaming algorithms bounds (given in~\Cref{table:prior}) are $\Omega(\sqrt{n})$.
% \footnote{We remark here that the claim is based on the 
% theoretical bounds reported by existing algorithms as given in~\Cref{table:prior}.}.

Our bound of $\tO(m\degen/\nt)$ subsumes the term 
$\tO(m^{3/2}/\nt)$, and dominates the term $\tO(m/\sqrt{\nt})$ when
$\nt = \Omega(\degen^2)$. For real-world graphs, $\nt = \Omega(\degen^2)$ is a naturally occurring phenomenon. In fact,
real-world large graphs are often characterized by following
two properties: (1) low sparsity, and (2) high triangle density~\cite{watts1998collective,sala2010measurement,seshadhri2012community,durak2012degree}. Thus, from a practical standpoint, our bound offers significant improvement over previously known bounds.

We complement Theorem~\ref{thm:main} with a nearly matching lower bound.

\begin{theorem} \label{thm:lb} Any constant pass randomized streaming algorithm for graphs with $m$ edges, $T$ triangles, and degeneracy at most $\kappa$, that provides a constant factor approximation to $T$
with probability at least $2/3$, requires storage $\Omega(m\degen/T)$.
\end{theorem}

We remark that all our results in this paper can be
equivalently stated  in terms of {\em arboricity} as
well. The {\em arboricity}
of a graph $G$, denoted as $\alpha$, is the smallest integer $p$ such that the edge set $E(G)$ can be partitioned into $p$ forests. It is asymptotically same as {\em degeneracy}: for every graph $\alpha \le \degen \le 2\alpha-1$. 

\subsection{Main ideas} \label{sec:ideas}

We give a high-level description of our algorithm and proof. The final algorithm
has a number of moving parts, and is based on recent advances in sublinear algorithms
for clique counting~\cite{ELRS,ERS18,ERS20}.

The starting point for our algorithm (and indeed, most triangle counting results related to degeneracy) is the classic sequential procedure of Chiba-Nishizeki. For every edge $e= (u,v)$, the size of intersection on neighborhoods of $u$ and $v$ is the number of triangles containing $e$. This intersection can be determined easily
in $\min(d_u,d_v)$ operations, by searching for elements of the smaller neighborhood in the larger one. (Here, $d_u$ denotes the degree of vertex $u$.) For convenience,
let us define the degree of edge $e$ to be $d_e := \min(d_u,d_v)$. Thus, we can enumerate all triangles in $\sum_{e}d_e$ time. The classic bound of Chiba-Nishzeki asserts that $\sum_{e} d_e$ $= O(m\degen)$.

As a warmup, let us get an $O(m\degen/T)$ space streaming algorithm, that uses a degree oracle. Define $d_E := \sum_e d_e = O(m\degen)$. With the degree oracle, in a single pass, we can sample an edge $e$ proportional to its degree. In the second pass, pick a uniform random neighbor $w$ of the lower degree endpoint of $e$. In the third pass, determine if $e$ and $w$ form a triangle. The probability of finding a triangle is exactly $3T/d_E$. By sampling $O(d_E/T) = O(m\degen/T)$ independent random edges in the first pass, we can estimate $T$ with $O(m\degen/T)$ space.

The main challenge is in removing the degree oracle. As a first step, can we effectively simulate sampling edges proportional to their degree? We borrow a key idea from recent sublinear algorithms for clique counting. First, we sample a set of uniform random edges, denoted $R$. In a second pass, we compute the degree of all edges in $R$. Now, we can run the algorithm described earlier, except we only sample edges of $R$. 
Observe that in the latter sample, taking expectations over $R$, we do sample edges proportional to their degree from the overall graph. Unfortunately, these samples are all correlated by the choice of $R$. How large should $R$ be to ensure that this simulation leads to the right answer?

An alternate viewpoint is to observe that the above approach can give an accurate estimate to the number of triangles incident to $R$, denoted $t_R$. We require $R$ to be large enough, so that $t_R$ can be used to estimate $T$. Let $t_e$ be the number of triangles incident to $e$. The $\{t_e\}$ values can exhibit large variance, even when $\degen = O(1)$. Consider a graph formed by $(n-2)$ triangles that all share a common edge. The graph is planar, so $\degen = O(1)$. But one edge is incident to $(n-2)$ triangles, and all other edges are incident to a single triangle. Thus, the $\{t_e\}$ values have the largest possible variance, and one cannot estimate $T$ by computing $\sum_{e \in R} t_e$ for a small $R$. Note that, for this example graph, our desired streaming algorithm uses only polylogarithmic space ($T = \Theta(m)$, $\degen = O(1)$).

Another key idea from sublinear clique counting saves the day: assignment rules. The idea is to assign triangles uniquely to edges, so that the distribution of assigned triangles has low variance. The overall algorithm will estimate the number of triangles \emph{assigned} to $R$, and not count the number of triangles \emph{incident} to $R$. A natural, though seemingly circular, rule is to assign each triangle to the contained edge that itself participates in the fewest triangles. Using properties of graph degeneracy, it is shown in~\cite{ERS20} that the maximum number of assigned triangles to any edge is $O(\degen)$. (Technically, this is not true. We have to leave some triangles unassigned.)

This leads to another technical complication. In the overall algorithm, when a triangle incident to an edge is discovered, the algorithm needs to determine if the triangle is actually assigned to the edge. This requires estimating $t_e$ for all edges $e$ in the triangle, a potentially space intensive operation. To perform this estimation in $O(m\degen/T)$ requires subtle modifications to the assignment procedure. It turns out we can ignore triangles containing edges of high degree, and thus, the above $t_e$ estimation is only required for low degree edges. Furthermore, we only need to determine if $t_e = \Omega(\degen)$,
which allows for smaller storage algorithms.

All in all, by choosing  parameters carefully, all steps can be implemented using $(m\degen/T)\poly(\log n,\eps^{-1})$ storage.
\section{Related Work}
\label{sec:related}

\begin{table*}[!ht]
\centering % used for centering table

\newlength{\myleftcolwd}\settowidth{\myleftcolwd}{Triangle counting}
\begin{tabular}{|c|c|c|}
\hline
{\bf Space}  & {\bf Remarks} & {\bf Source} \\
\hline
% \multirow{18}{\myleftcolwd}{Triangle counting ($\TRIC$)}
% & 
\Tstrut $\tO\big({mn}/{\nt}\big)^2$ & one pass & \cite{BarYossefKS02}  \\ 
 $\tO\big({m \dmax^2}/{\nt}\big)$ & one pass, $\dmax =$ maximum degree & \cite{Jowhari2005}     \\ 
 $\tO\big({m n}/{\nt}\big)$ & one pass, $n$ known a priori & \cite{Buriol2006}     \\ 
 $\tO\big({m^3}/{\nt^2}\big)$ & one pass, dynamic stream  & \cite{Kane2012}     \\ 
 $\tO\big({m \dmax}/{\nt}\big)$ & one pass, $\dmax =$ maximum degree & \cite{Pavan2013}     \\ 
%  \Bstrut $\tO\big({m \tangle}/{\nt}\big)$ & $\tangle =$ tangle coefficient (Definition~\ref{def:Tangle}) & \cite{Pavan2013}     \\ 
  \Bstrut $\tO\big({m \temax}/{\nt} +  m /{\sqrt{\nt}} \big)$ & one pass,  $\temax =$ maximum triangles  incident on a edge & \cite{Pagh2012}     \\ 
 \Bstrut $\vertexcover + \tO\big( \pathcount /{\nt}\big)$ & one pass,  $\vertexcover =$ vertex cover, $\pathcount =$ \# of $2$-paths & \cite{garcia2014}  \\  
 \Bstrut $\tO\big( m/\sqrt{T}\big)$ & dependence on $\eps$ is $1/\eps^{2.5}$ & \cite{cormode2014}  \\  
 \Bstrut $\tO \big( m^{3/2} /{\nt}\big)$ & multi-pass & \cite{McGregor2016,bera2017}     \\ 
 \Bstrut $\tO \big( m /{\sqrt{\nt}}\big)$ & multi-pass & \cite{McGregor2016}     \\  \cline{1-3}
 \Tstrut $\Omega\big( n^2\big)$ &  one pass, $\nt=1$ & \cite{BarYossefKS02}     \\ 
 $\Omega\big( n/\nt\big)$ & multi-pass, $\nt<n$  & \cite{Jowhari2005}     \\ 
 $\Omega\big( m\big)$ & one pass, $m\in [c_1 n, c_2 n^2],~\nt<n$ & \cite{Braverman2013}     \\ 
 $\Omega\big( m/\nt\big)$ & multi-pass & \cite{Braverman2013}     \\ 
 \Bstrut $\Omega\big( m^3/\nt^2\big)$ & one pass, optimal &    \cite{Kutzkov2014}  \\ 
 \Bstrut $\Omega\big( m / \nt^{2/3} \big)$ & multi-pass &    \cite{cormode2014}  \\ 
 \Bstrut $\Omega\big( m / \sqrt{\nt}\big)$ & multi-pass, for $m= \Theta(n\sqrt{T})$  &    \cite{cormode2014}  \\
 \Bstrut $\Omega\big( \min \{m / \sqrt{\nt} , m^{3/2}/{\nt} \}\big)$ & multi-pass   &    \cite{bera2017}  \\
\hline
\end{tabular}
\caption{Prior work on the triangle counting problem}

\label{table:prior} % is used to refer this table in the text
\end{table*}

The triangle counting problem, a special case of more general subgraph counting problem, has been studied extensively in the 
streaming setting. We present a summary of the significant prior 
works in~\Cref{table:prior}. The upper bounds stated in the table
are for randomized streaming algorithms that provide $(1\pm \eps)$-approximation to the true triangle count with probability at least $2/3$. The $\tO$ notion hides polynomial dependencies on $1/\eps$ and $\log n$. The lower bounds are primarily based on the triangle detection
problem --- detect whether the input graph is triangle free or it contains at least $T$ many triangles. All the results presented
in the table are for the arbitrary order stream.

Jha \etal~\cite{Jha2013} designed a one pass 
$\tO(m/ \sqrt{\nt})$-space algorithm with $\pm W$-additive error 
approximation, where $W$ is the number of two length paths (also called wedges). Braverman\etal~\cite{Braverman2013}
gave a two-pass $\tO(m/\nt^{1/3})$-space algorithm to detect
if the input graph is triangle free or it has at least $T$
many triangles. These results are not directly comparable to our work.

The triangle counting problem has been studied in the context of the
adjacency list streaming model as well. This model is
also known as the vertex arrival model: all the edges incident 
on a vertex arrive together. McGregor~\etal~\cite{McGregor2016} gave one-pass $\tO(m/\sqrt{\nt})$-space and two pass $\tO(m^{3/2}/{\nt})$-space algorithm for the triangle counting problem
in this model. We refer to~\cite{McGregor2016} for other
related work in this model.

Bounded degeneracy graph family is an important class of 
graphs from a practical point of view. Many real-world large
graphs, specially from the domain of social networks and web graphs, often exhibit low degeneracy(~\cite{goel2006bounded,JS17,shin2018patterns,DBS18,bera2019graph}, also Table 2 in~\cite{bera2019graph}). Naturally, designing algorithms 
that are parameterized by {\em degeneracy} has been a theme
of many works in the streaming settings; some examples include
matching size estimation~\cite{assadi2017estimating,esfandiari2018streaming,cormode2017sparse}, independent
set size approximation~\cite{cormode2017independent}, graph coloring~\cite{bera2019graph}. In the general RAM model, 
the relation between degeneracy and subgraph counting
problems has been explored in~\cite{Chiba1985,eppstein1994arboricity,bera2020}.

In the graph query model, where the goal is to design 
sub-linear time algorithms, Eden~\etal~\cite{eden2018faster}
studied the triangle counting problem, and more generally the clique counting
problem in {\em bounded degeneracy} graphs. Although the model is significantly different from the streaming model, we port some 
key ideas from there; see~\Cref{sec:ideas} for a detailed discussion.
The relevance of {\em bounded degeneracy} has been further 
explored in the context of estimating degree moments~\cite{eden2017sublinear} in this model.

%Add more practical examples here.

\section{Notations and Preliminaries}

For an integer $k$, we denote the set $\{1,2,\ldots,k\}$ by $[k]$. 
Throughout the paper, we denote the input graph
as $G=(V,E)$. We assume $G$ has $n$ vertices, $m$ edges 
and $\nt$ many triangles. We denote the degree of a vertex $v\in V$ by $d_v$ and its neighborhood by $N(v)$. For an edge $e=\{u,v\}$, we define its neighborhood $N(e)$ to be that of the lower degree end point: $N(e) = N(u)$ if $d_u < d_v$; $N(e) = N(v)$ otherwise.
Similarly, we define the degree of an edge:
$d_e = \min \{d_u,d_v\}$. 
For a collection of edges $R$, we define $d_R =\sum_{e\in R}d_e$.
In particular, $d_E := \sum_{e\in E}d_e$. 

% A graph $G$ is called $k$-degenerate if every induced subgraph of $G$ has a vertex of degree at most $k$. The {\em degeneracy}  $\degen_G$ is the smallest integer $k$ such that $G$ is $k$-degenerate. The {\em arboricity} $\alpha_G$ is the smallest integer $p$ such that the edge set $E(G)$ can be partitioned into $p$ forests. When the graph $G$ is clear from the context, we drop the subscript $G$ and simply use $\degen$ and $\alpha$ instead. The {\em degeneracy} and {\em arboricity} are closely related concepts: for every graph $\alpha \le \degen \le 2\alpha-1$. In this paper, we choose to state our results in terms of {\em degeneracy}; they can be equivalently stated 
% in terms of {\em arboricity} as well.

Chiba and Nishizeki~\cite{Chiba1985} proved the
following insightful connection between the sum of degrees of
the edges in a graph $d_E$ and its {\em degeneracy} $\degen$.
\footnote{Note that Chiba and Nishizeki~\cite{Chiba1985} stated
their results in terms of {\em arboricity}. 
As $\alpha \le \degen$ for each graph $G$ with 
{\em arboricity} $\alpha$, the same result holds 
with respect to {\em degeneracy} $\degen$ as well.
}

\begin{lemma}[Lemma 2 in~\cite{Chiba1985}]
\label{lem:chiba}
For a graph $G$ with $m$ edges and {\em degeneracy} $\degen$, 
$$ d_E = \sum_{e\in E} d_e \leq  2m \degen \,.$$
\end{lemma}
As a corollary, we get the following result.
\begin{corollary}[~\cite{Chiba1985}]
\label{lem:max_tri}
For a graph $G$ with $m$ edges and {\em degeneracy} $\degen$, the maximum number of triangles in $G$
is at most $2m\degen$.
\end{corollary}

% We will make use of {\em degree based vertex ordering},
% denoted by $\order$, in our algorithms. In this
% ordering, $u \order v$ iff $d_u < d_v$ or $d_u=d_v$ and $u$
% appears before $v$ in some arbitrary but fixed (say
% lexicographical) ordering of the vertices.

We use the notation $\tO(~\cdot~)$ to hide polynomial dependencies on 
$(1/\eps)$ and $\log n$ terms, where $\eps$ is the error
parameter. For designing our algorithms, we focus on the expected space usage. This can be easily converted
into a worst-case guarantee by applying Markov inequality
--- simply abort if the space usage runs beyond $c$ times the expected space usage, for some constant $c$. This only increases
the error probability by an additive $1/c$ amount.

We use the following variants of the Chernoff bound and Chebyshev inequality for analyzing our algorithms.
\begin{theorem}[Chernoff Bound~\cite{Ch52}] 
\label{thm:chernoff}
  Let $X_1,X_2,\ldots,X_r$ be mutually independent indicator random variables with expectation $\mu$. Then, for every $\eps$ with $0< \eps <1$, we have
  \begin{align*}
       \Pr\left[ \Big\lvert \frac{1}{r}\sum_{i=1}^{r} X_i - \mu \Big\rvert \ge \eps \mu \right] \le 2\exp\left( - \eps^2 r \mu  / 3 \right)
  \end{align*}
\end{theorem}
\begin{theorem}[Chebyshev Inequality~\cite{Alsmeyer2011}]
\label{thm:chebyshev}
Let $X$ be a random variable with expectation $\mu$
and variance $\Var[X]$. Then, for every $\eps>0$,
\begin{align*}
    \Pr\left[ \lvert X - \mu \rvert \geq \eps \mu\right] \leq \frac{\Var[X]}{ \eps^2 \mu^2} \,.
\end{align*}
\end{theorem}

\section{Warm-up: An abstract model}
\label{sec:warm_up}

In this section, we consider a streaming model
equipped with a degree oracle: queried with a vertex $v$, 
the oracle returns $d_v$. Furthermore, we make a rather strong
assumption: there is no cost associated with the queries. 
McGregor~\etal~\cite{McGregor2016} designed a $\tO(m^{3/2}/\nt)$ space $3$-pass streaming algorithm
in this model --- their algorithm makes $O(m)$ many
degree queries. We describe an $\tO(m\degen/\nt)$-space 
$3$ pass algorithm in this model. Our estimator makes $2m$ many degree queries and requires $3$-pass. For {\em bounded
degeneracy} graph families, this translates to a space reduction by a factor of $O(\sqrt{m})$. 
In the next section, we show how to design
a $\tO(m\degen/\nt)$-space constant pass algorithm in the
traditional streaming model.

Our main idea is to sample edges from the stream with 
probability proportional to its degree. In general streaming settings,
this is not possible as we do not know the degree 
of the edges apriori. However, the model that we consider here is tailor-made for this purpose. It shows the effectiveness
of degree-biased edge samples in estimating triangle count
and provides motivation for taking up a similar
sampling approach in the general streaming model.

We present our basic estimator in~\cref{alg:triangle_estimate_ideal}. In the full 
algorithm, we will run multiple instances of this
estimator in parallel and report the ``median
of the mean''~\cite{Chakrabartics49} as our final estimate.
%We show that our estimator is unbiased, and has bounded variance..

\begin{algorithm}[!ht]
\caption{A Triangle Estimator }
\label{alg:triangle_estimate_ideal}
    \begin{algorithmic}[1]
    \Procedure{IdealEstimator}{Graph $G=(V,E)$} 
    \label{proc:ideal_estimator}
    \State Pass 1: Sample an edge $e$ with probability $d_e/d_E$.
    \State Pass 2: Sample a vertex $w$ from $N(e)$ u.a.r.
    \State Pass 3: Check if $\{e,w\}$ forms a triangle.
    \If {$\tau = \{e,w\}$ is a triangle} 
    \State Call \textsc{IsAssigned}$(\tau, e)$. 
    \State If returned YES, then set $Y=1$; else set $Y=0$.
    \Else
    \State Set $Y=0$.
    \EndIf
    \State Set $X = d_E \cdot Y$.
    \State return $X$.
    \EndProcedure
    % \Procedure{IsAssigned}{ triangle $\tau$, edge $e$ } 
    % \label{proc:ideal_estimator}
    % \State Let $\tau=\{u,v,w\}$ and $e=\{u,v\}$.
    % \If {$u \prec w$ and $v\prec w$} 
    % \State return YES
    % \Else 
    % \State return NO
    % \EndIf
    % \EndProcedure
    \end{algorithmic}    
\end{algorithm}

\mypar{Implementation Details} 
The degree proportional sampling is achieved by using 
weighted reservoir sampling~\cite{chao1982general}. On arrival of the
edge $e=\{u,v\}$ in the stream, we make two degree
queries to find $d_e$.
The method \textsc{IsAssigned} is required to 
ensures that every triangle is uniquely associated with
one of its three edges. Other than this, there
is no constraints on the implementation of 
this method. For example, we can associate every triangle
to the edge with lowest degree, breaking ties arbitrarily
(but consistently). Let $t_e$ denote the number 
of triangles assigned to the edge $e$.
Clearly, $\sum_{e \in E} t_e = {\nt}$.
% that every
% triangle is associated with an unique edge-vertex pair. 
% We use {\em degree based vertex ordering} and 
% assign each triangle $\tau= {e,w}$ to the edge $e$
% such that the remaining vertex $w$ succeeds both
% the end points of $e$ in the ordering. 
% We now turn to analyze the above estimator.

\mypar{Analysis}
First, we show the estimator is unbiased.
\begin{align*}
    \EE[X] &= \sum_{e \in E} \frac{d_e}{d_E} \cdot \EE[X|e] \\
            & = \sum_{e \in E} {d_e} \cdot  \EE[Y|e] \\
            &= \sum_{e \in E} {d_e} \cdot  \frac{t_e}{d_e} 
            = \sum_{e \in E} t_e 
            = \nt
\end{align*}
Now we bound the variance of the estimator.
\begin{align*}
    \Var[X] \leq \EE[X^2] 
            &= \sum_{e \in E} \frac{d_e}{d_E} \cdot \EE[X^2|e] \\
            &= \sum_{e \in E} d_e \cdot d_E \cdot  \EE[Y^2|e] \\
            &= \sum_{e \in E} {d_e} \cdot d_E \cdot   \frac{t_e}{d_e} \\
            &= d_E \cdot \sum_{e \in E} t_e 
            = d_E \cdot {\nt}
\end{align*}
So, running $\tO(\Var[X]/\EE[X]^2)=\tO(d_E/\nt)=\tO(m\degen/\nt)$-many estimators independently in parallel suffices for a $(1\pm\eps)$-approximate estimate. Since each copy of the estimator
requires constant space, the overall space usage is 
bounded by $\tO(m\degen /\nt)$.
\section{Our Main Algorithm}
In this section, we present our streaming triangle estimator.
As promised, our algorithm does not assume access to a degree oracle. If the model is indeed equipped with a degree oracle, then we can save a few passes over the stream. Perhaps more importantly, the number of queries to the oracle is upper bounded by the space usage of our algorithm. Our main algorithmic result is the following.
\begin{theorem} 
\label{thm:main2} Consider a graph $G$ of degeneracy at most $\degen$,
that is input as an arbitrary edge stream. There is a streaming algorithm that 
outputs a $(1\pm\eps)$-approximation to $T$, with high probability\footnote{We use ``high probability" to denote errors less than $1/3$.}, and has the following properties.
It makes six passes over the input stream and uses space $(m\degen/T)\cdot \poly(\log n, \eps^{-1})$.
\end{theorem}

We describe our estimator in~\Cref{alg:triangle_estimate}.
We set the parameters $r$ and $\ell$ later in the analysis. 
In analyzing our algorithm, the procedure \textsc{IsAssigned}
would play a crucial role. As discussed earlier,
\textsc{IsAssigned} takes as input a triangle and an edge,
and outputs whether the triangle is assigned to that edge.
We want this procedure to possess four properties: 
(1) For a given triangle, it is either unassigned or assigned uniquely to 
one of its three participating edges.
(2) {\em almost} all the triangles are assigned,
(3) For any fixed edge, not too many triangles are assigned to it, and
(4) The space complexity of the procedure is bounded by $\tO({m\degen}/{{\nt}})$.
The first two properties are required to ensure the overall accuracy of the estimator. The third property would be central to bounding the variance of the estimator. The final property will ensure
the overall space complexity of our triangle estimator is bounded 
by $\tO({m\degen}/{\nt})$.

We analyze our triangle
estimator assuming a black-box access to a \textsc{IsAssigned} procedure
that satisfies the above four properties. In the next section, we will 
take up the task of designing such an assignment procedure in the 
streaming setting. We make the above discussion rigorous and formal below.

Assume $\assign_e$ denotes the number of triangles assigned to
the edge $e$ by the procedure \textsc{IsAssigned}.
We use $\assign_{\max}$ to denote the maximum number 
of triangles that any edge has been assigned to:
$\assign_{\max} = \max_{e\in E} \assign_e$.
Denote the number of triangles that \textsc{IsAssigned} assigns to some
edge by $\total = \sum_{e\in E} \assign_e$.
Then, for any positive constant $\eps$ and $\delta$,
we define an $(\eps,\delta)$-accurate \textsc{IsAssigned}
procedure below.
\begin{definition}[$(\eps,\delta)$-accurate \textsc{IsAssigned}]
\label{def:assign}
A procedure \textsc{IsAssigned} that assigns a triangle
to an edge or leaves it unassigned,
is $(\eps,\delta)$-accurate if it satisfies the following four
properties.
\begin{enumerate}
    \item {\em Unique Assignment:} For each triangle, it is either unassigned or uniquely assigned to one of the three participating edges. This implies,
    $\total \leq \nt$.
    \item {\em Almost All Assignment:} With probability at least $1-\delta$, $\total \geq (1-12\eps){\nt}$.
    \item {\em Bounded Assignment:} With probability at least $1-\delta$, $\assign_{\max} \leq \degen/\eps$.
    \item {\em Bounded Space Complexity:} Each call requires 
    $\tO\left( {m\degen}/{\nt} \right)$ bits of space.
\end{enumerate}
\end{definition}

We now analyze our algorithm assuming a black-box access to $(\eps,O(1/n^5))$-accurate \textsc{IsAssigned}. The analysis consists of two parts.
First, we show that for a certain settings of $r$ and $\ell$,
our final estimate $X$ is indeed 
a $(1\pm \eps)$ approximation to the true triangle count.
In the sequel, we bound the space complexity of our algorithm. 

\begin{algorithm}[!ht]
\caption{Estimation of triangle count}
\label{alg:triangle_estimate}
    \begin{algorithmic}[1]
    \Procedure{EstimateTraingle}{Graph $G=(V,E)$}
    \State Pass 1: Sample $r$ many edges  u.a.r: $R = \{e_i\}_{i=1}^r$.
    \State Pass 2: Compute $d_e$ for each $e\in R$.
    \For {$i=1$ to $\ell$} \label{line:loop}
    \State Sample an edge $e \in R$ independently with prob. $d_e/d_R$.
    \State Pass 3: Sample a vertex $w$ from $N(e)$ u.a.r.
    \State Pass 4: Check if $\{e,w\}$ forms a triangle.
    \If {$\tau = \{e,w\}$ is a triangle} \label{line:if}
    \State Call \textsc{IsAssigned}$(\tau, e)$. 
    \State If returned YES, then set $Y_i=1$; else set $Y_i=0$.
    \Else
    \State Set $Y_i=0$.
    \EndIf
    \EndFor
    \State Set $Y = \frac{1}{\ell} \sum_{i=1}^{\ell} Y_i$, and 
    $X = \frac{m}{r} \cdot d_R \cdot Y$. \label{line:Y}
    \State return $X$.
    \EndProcedure
    % \Procedure{EdgeSample}{Graph $G=(V,E)$, integer $r$} \label{proc:edge_sample}
    % \For {$i=1$ to $r$}
    % \State sample an edge $e_i\in E$ independently and u.a.r.
    % \EndFor
    % \State Let $R=\{e_i\}_{i=1}^{r}$,
    % $d_R = \sum_{i=1}^{r} d_{e_i}$, and $t_R = \sum_{i=1}^{r} t_{e_i}$.
    % \Comment{We do not estimate $t_R$}
    % \State return $R$.
    % \EndProcedure
    \end{algorithmic}
\end{algorithm}

We begin with analyzing the {\em quality}
of the (multi)set of uniform random edges $R$. Collectively 
through~\Cref{lem:set_R,def:good_R,lem:good_set}, we 
establish that for a suitable choice of the parameter $r$,
the (multi)set $R$ possesses desirable properties with high
probability.
\begin{lemma}
\label{lem:set_R}
Let $d_R = \sum_{e \in R} d_e$ and $\assign_R = \sum_{e \in R} \assign_e$.
For any constant $\eps>0$ we have
\begin{enumerate}
    \item $\EX[d_R] = r \cdot \frac{d_{E}}{m} $ and $\EX [\assign_R] = r\cdot  \frac{\total}{m} $,
    \item $\Pr \left[d_R \leq \EX \left[ d_R \right] \cdot \frac{\log n}{\eps}\right] \geq 1 -\frac{\eps}{\log n}$,
    \item $\Pr \left[|\assign_R -\EX[t_R]| \leq \eps \EX[\assign_R] \right] \geq 1 - \frac{1}{\eps^2} \cdot \frac{1}{r} \cdot \frac{m \cdot  \assign_{\max} }{\total} $.
\end{enumerate}
\end{lemma}
\begin{proof}
We first compute the expected value of $d_R$ and $\assign_R$. 
We define two sets of random variables, $Y_{i}^{d}$ and $Y_{i}^{t}$ for $i\in [r]$ as follows:
$Y_{i}^{d} = d_{e_i}$, and $Y_{i}^{t} = \assign_{e_i}$. Then, $d_R = \sum_{i=1}^{r} Y_{i}^{d}$ and $\assign_R = \sum_{i=1}^{r} Y_{i}^{t}$. We have
\begin{align*}
    \EX \left[Y_{i}^{d} \right] &= \sum_{e \in E} \Pr[e_i=e] \cdot \EX 
    \left[ Y_{i}^{d} | e_i = e\right] \\
    &= \frac{1}{m} \sum_{e \in E} d_e \\
    &= \frac{d_{E}}{m} \,.
\end{align*}
Then, by linearity of expectation, we get $\EX [d_R] = r \cdot {d_{E}}/{m} $.
Analogously, we have $\EX [\assign_R] = r\cdot  \total /m $.

We now turn our focus on the concentration of $d_R$. This is achieved by a simple application of Markov inequality.
\begin{align*}
    \Pr \left[ d_R \geq \EX \left[ d_R \right] \cdot \frac{\log n}{\eps} \right] \leq \frac{\eps}{\log n} \,.
\end{align*}

To prove a concentration bound on $\assign_R$, we study the variance of $\assign_R$.
By independence, we have 
\begin{align*}
    \Var[\assign_R] 
    &= \sum_{i=1}^{r} \Var [Y_i^t] \\
    &\leq  \sum_{i=1}^{r} \EX [(Y_i^t)^2] \\
    &= \sum_{i=1}^{r} \sum_{e \in E} \Pr [e_i = e]~\EX [(Y_i^t)^2 | e_i=e] \\
    &= \frac{r \cdot \sum_{e \in E} \assign_e^2}{m}  \\
    &\leq  \frac{r \cdot \assign_{\max} \sum_{e \in E} \assign_e}{m}  \\
    &= \frac{r \cdot \assign_{\max} \total }{m}\,.
\end{align*}
Then, the item (3) of the lemma follows by an application of Chebyshev inequality (~\Cref{thm:chebyshev}).
\end{proof}
We next define a collection of edges $R$ as {\em good}
if the conditions in items (2) and (3) in~\Cref{lem:set_R} are satisfied. Formally, we have the following definition.
\begin{definition}[A {\em good} collection of edges]
\label{def:good_R}
We call a fixed collection of edges $R$, $\eps$-good if the following two conditions are true.
\begin{align}
\label{eqn:goodR}
d_R &\leq \frac{\log n}{\eps} \cdot |R| \cdot \frac{d_E}{m}  \\
\assign_R &\in \left[ (1-\eps)\cdot |R|\cdot \frac{\total}{m}~~,~~ (1+\eps)\cdot |R|\cdot \frac{\total}{m}\right] 
\end{align}
\end{definition}
\begin{lemma}[Setting of $r$ for an $\eps$-good R]
\label{lem:good_set}
Let $0< \eps < 1/6$ and $c>6$ be some constants, and $r = \frac{c \log n}{\eps^2} \frac{m \assign_{\max}}{\total}$. Then, with probability at least $1-\frac{1}{6\log n}$, $R$ is $\eps$-good.
\end{lemma}
\begin{proof}
The lemma follows by plugging in $r = \frac{c \log n}{\eps^2} \frac{m \assign_{\max}}{\total}$ in~\Cref{lem:set_R}
and using the bounds on $c$ and $\eps$.
\end{proof}

We have established that the random collection of edges $R$ is 
{\em good} with high probability.
We now turn our attention to the random variable $Y$, as 
defined on~\Cref{line:Y} ~\Cref{alg:triangle_estimate}.
Together in~\Cref{lem:triangle_est,lem:goodY} we show 
that, if $R$ is {\em good} then for a suitably chosen parameter $\ell$, the random variable $Y$ is well-concentrated around its mean.

\begin{lemma}
\label{lem:triangle_est}
Let $R$ be a fixed collection of edges, and $Y_R$ denote the value of the random variable $Y$ as defined on~\Cref{line:Y} ~\Cref{alg:triangle_estimate} on $R$.  Then,
\begin{enumerate}
    \item $\EX [Y_R] = \frac{\assign_R}{d_R}$,
    \item $\Pr \left[~|Y_R - \EX[Y_R]| \geq \eps \EX [Y_R]~ \right] 
    \leq \exp \left( - \ell \cdot \frac{\eps^2}{3} \cdot \frac{\assign_R}{d_R}  \right)$.
\end{enumerate}
\end{lemma}
\begin{proof}
Let $e_i$ be the edge sampled in the $i$-th iteration of the for loop at ~\cref{line:loop} in ~\Cref{alg:triangle_estimate}. Then,
\begin{align*}
    \EE[Y_i=1] &= \sum_{e\in R} \Pr[e_i = e]\Pr [Y_i =1 | e_i = e] \,,  \\
    &= \sum_{e\in R} \frac{d_e}{d_R} \Pr [Y_i =1 | e_i = e] \,,  \\
    & = \sum_{e \in R} \frac{d_e}{d_R}\cdot  \frac{\assign_e}{d_e} \\
    &= \sum_{e \in R} \frac{\assign_e}{d_R}\\
    &= \frac{\assign_R}{d_R} \,.
\end{align*}
By linearity of expectation, we have the item (1) of the lemma.
For the second item, we apply Chernoff bound(~\Cref{thm:chernoff}).
\end{proof}
\begin{lemma}[Setting of $\ell$ for concentration of $Y_R$]
\label{lem:goodY}
Let $0 < \eps < 1/6$ and $c>20$ be some constants, and $\ell = \frac{c \log n}{\eps^2}\cdot \frac{m \cdot d_R}{r \cdot \total}$. Then, 
with probability at least $1- \frac{1}{5\log n}$, $|Y_R - \EX[Y_R]| \leq \eps \EX[Y_R] $.
\end{lemma}
\begin{proof}
By~\Cref{lem:good_set}, $R$ is $\eps$-good with probability at least
$1-\frac{1}{6\log n}$. Condition on the event that $R$ is $\eps$-good. 
By definition of a $\eps$-good set, $\assign_R$ is tightly concentrated
around its mean: $\assign_R \in \left[ (1-\eps){r\total}/{m}~~,~~ (1+\eps) {r\total}/{m}\right] $
Then, by item 2. in~\cref{lem:triangle_est}, we have
\begin{align*}
    &\Pr \left[~|Y_R - \EX[Y_R]| \geq \eps \EX [Y_R]~ \right] \\
    & \leq \exp \left( - \frac{c \log n}{\eps^2}\cdot \frac{m  d_R}{r  \total} \cdot \frac{\eps^2}{3} \cdot \frac{\assign_R}{d_R}  \right) \\
    &\leq \exp \left( - \frac{c \log n}{3}\cdot \assign_R \cdot \frac{m}{r\total} 
     \right) \\
    &= o(1/n^3)\,,
\end{align*}
where the last line follows from the concentration of $\assign_R$ for
$\eps$-{good} $R$.
Removing the condition on $R$, we derive the lemma.
\end{proof}

We have now all the ingredients to prove that our
final estimate is indeed close to the actual triangle count.
The random variable $Y$ is scaled appropriately to
ensure that its expectation is close to the true triangle
count.

\begin{lemma}
\label{lem:tri} Assume $r$ and $\ell$ is set 
as in~\Cref{lem:good_set} and ~\Cref{lem:goodY} respectively.
Then, there exists a small constant $\eps'$ such that with probability at least $1-\frac{1}{3\log n}$, $X \in [(1-\eps' \nt), (1+\eps' \nt)]$. 
\end{lemma}
\begin{proof}
With probability at least $1-\frac{1}{5\log n}$, $Y_R$ is closely
concentrated around its expected value $\assign_R/d_R$ (by~\Cref{lem:goodY}).
More formally,
\begin{align*}
    Y_R \in \left[(1-\eps) \frac{\assign_R}{d_R}, (1+\eps) \frac{\assign_R}{d_R} \right] \,.
\end{align*}
Then, with high probability
\begin{align*}
    X_R \in \left[(1-\eps)\cdot \frac{m}{r} \cdot {\assign_R}, (1+\eps) \cdot \frac{m}{r} \cdot {\assign_R} \right] \,.
\end{align*}
Since $R$ is good, with probability at least $1-\frac{1}{c\log n}$, 
\begin{align*}
    t_R \in \left[ (1-\eps)r\cdot \frac{\total}{m}, (1+\eps)r\cdot \frac{\total}{m},  \right]
\end{align*}
Then, with probability at least $1-\frac{2}{c\log n}$,
\begin{align*}
    X_R \in \left[ (1-2\eps) \total, (1+2\eps) \total  \right]
\end{align*}
Removing the conditioning on $R$, and using the bound on $\total$
as given in~\Cref{def:assign}, 
we have $$X \in \left[ (1-\eps') \nt, (1+\eps') \nt  \right]$$ with probability at least $1-\frac{4}{c\log n}$, for suitable chosen 
parameter $\eps'$.
\end{proof}

This completes the first part of the analysis.
We now focus on the space complexity of~\Cref{alg:triangle_estimate}.
\begin{lemma}[Space Complexity of~\Cref{alg:triangle_estimate}]
\label{lem:space}
Assuming an access to a $(\eps,\delta)$-accurate \textsc{IsAssigned} method,
~\Cref{alg:triangle_estimate} requires $\tO(m\degen/{\nt})$ bits of storage
in expectation.
\end{lemma}
\begin{proof}
Clearly, $O(r+\ell)$ space is sufficient to 
store the set $R$ and sample $\ell$ many edges from it at
~\Cref{line:loop}. Recall from~\Cref{lem:good_set,lem:goodY}
that $r = \tO(m\assign_{\max}/ {\total})$ and $\ell = \tO(md_R/(r\total))$,
respectively. Using the bound on $\total$ and $\assign_{\max}$
from the definition of $(\eps,\delta)$-accurate \textsc{IsAssigned}
method, we derive that $O(r+\ell)$ is $\tO(m\degen/\nt)$ with high probability.
 
We now account for the space complexity of the \textsc{IsAssigned} method. 
It is called if the the edge-vertex pair
$\{e,w\}$ forms a triangle (if condition at~\Cref{line:if}).
Let $Z_i$ be an indicator random variable to denote if \textsc{IsAssigned} 
is called during the $i$-th iteration of the
for loop at~\Cref{line:loop}. Then,
$$ \Pr[Z_i=1 |R] = \sum_{e \in R} \frac{d_e}{d_R} \cdot \frac{t_e}{d_e}  = \frac{t_R}{d_R} \,.$$
Then, the expected number of calls to \textsc{IsAssigned} is bounded
by $\ell \cdot \frac{t_R}{d_R} = \tO\left( \frac{m~t_R}{r~\total} \right)$.
Note that $t_R$ can be much different from $\assign_R$,
as it counts the exact number of triangles per edge. However, 
we show that with constant probability, $t_R$ is at most $O(r\total/m)$, which
bounds the expected number of calls by $\tO(1)$. Since each call to \textsc{IsAssigned}
takes $\tO(m\degen/\nt)$ bits of space, the lemma follows. We now bound $t_R$.
$$ \EE [t_R] = r \sum_{e \in E} \frac{t_e}{m} = \frac{3r\nt}{m} $$,
where the last equality follows from the fact that $\sum_{e\in E} t_e = 3\nt$.
An application of Markov inequality bounds the probability that $t_R$ is
more than $\frac{cr\total}{m}$ by a small constant probability, for any large constant $c>10$.
\end{proof}

Thus, assuming an access to a $(\eps,o(1/n^5))$-accurate \textsc{IsAssigned}
method,~\Cref{lem:tri,lem:space} together prove our
main result in~\Cref{thm:main2}.

\subsection{Assigning triangles to edges}
% In the previous section, we demonstrated that
% the space complexity of the triangle counting 
% problem is dictated by the edges that take part
% in large number of triangles. To reduce the 
% variance of the estimation, it is crucial to 
% ensure we do not charge too many triangles
% on each edge. A particular popular assignment
% rule is based on degree ordering of the vertices:
% map every triangle to the edge between two
% lower degree end-points (breaking ties arbitrarily,
% but consistently)~\cite{}. While such methods are
% effective in reducing the variance of the estimator,
% they do not exploit bounded degeneracy properties
% of the graph. In this section, we carefully
% design an assignment rule that exploits the 
% low degeneracy of the graph and results in an
% $O(m\degen)/{\nt}$-estimator for the triangle counting
% problem.

% Let $\nt_3$ denote the set of triangles in $G$.
% We define a partial function $\pi: \mathcal{T} \rightarrow 
% E$, that maps a triangle in $\mathcal{T}$ to an edge $E$.
% Note that $\pi$ is a partial function
% and will not map some triangles to any edges. 
% Let $t_{e,\pi}$ be the number of triangles
% that are mapped to the edge $e$ by $\pi$: $t_{e,\pi} = |\{\tau \in \mathcal{T}: \pi(\tau) = e\}|$.
% In defining the partial function $\pi$,
% we aim to satisfy the following two properties for some
% suitably chosen small constant $\eps$:
% \begin{itemize}
%     \item For every edge $e$, $t_{e,\pi} \leq  \degen / \eps$. 
%     \item Almost all the triangles are associated with some edges: $\sum_{e\in E} t_{e,\pi} \geq (1-2\eps){\nt}$. 
% \end{itemize}

In this section we give an algorithm for the \textsc{IsAssigned}
procedure in~\Cref{alg:IsAssigned}. Recall from the previous section
that we require \textsc{IsAssigned} to be $(\eps,\delta)$-accurate (see~\Cref{def:assign}).

The broad idea is to assign a triangle to the edge with smallest $t_e$.
Recall that $t_e$ is the number of triangles that the edge $e$ participates in.
However, computing $t_e$ might be too expensive in terms
of space required for certain edges. As evident from the analysis of~\Cref{alg:triangle_estimate},
we have a budget of $\tO(m\degen/{\nt})$ in terms of bits of storage for 
each call to \textsc{IsAssigned}. In this regard, we define ``heavy''
and ``costly'' edges and it naturally leads to a notion of ''heavy'' and ''costly'' triangles. If a triangle is either ``heavy'' or 
``costly'', then we do not attempt to assign it to any of its edges. 
Crucially, we show that
the total number of ``heavy'' and ``costly'' triangles are only a tiny fraction
of the total number of triangles in the graph. 
% Furthermore, if a triangle
% is neither ``heavy'' nor ``costly'', then we can estimate $t_e$ for each of its
% edges and assign is to the smallest one. We make these notions
% rigours in the following.

We need to ensure that for any edge $e$, not too many triangles are
assigned to $e$ by \textsc{IsAssigned}. We achieve this by simply
disregarding any edge with large $t_e$ from consideration while assigning 
a triangle to an edge. Formally we capture this by defining {\em heavy} edges and
triangles.
\begin{definition}[{\em $\eps$-heavy edge} and {\em $\eps$-heavy triangle }]
\label{def:heavy}
An edge is defined $\eps$-{\em heavy} if $t_e > \degen/\eps$.
A triangle is deemed $\eps$-{\em heavy} if {\em all} the three of 
its edges are $\eps$-heavy.
\end{definition}

If the ratio $t_e/d_e$ is quite small for an edge $e$, then we need too many samples from the neighborhood $N(e)$ to estimate $t_e$. Roughly speaking,
$O(d_e/t_e)$ many samples are required for an {\em accurate} estimation. In this regard, we define 
{\em costly edges} and {\em costly triangles} as follows.
\begin{definition}[{\em $\eps$-costly edge} and {\em $\eps$-costly triangle}]
\label{def:costly}
An edge $e$ is defined {\em $\eps$-costly} if $d_e/t_e > m\degen /(\eps{\nt})$. A triangle is deemed {\em $\eps$-costly} if {\em any} of
its three edges is {\em $\eps$-costly}.
\end{definition}

We first show that the number {\em heavy triangles} and 
{\em costly triangles} are only a small fraction of the
all triangles. Formally, we  prove the following lemma.

\begin{lemma}
\label{lem:heavy_costly_triangle}
The number of $\eps$-{\em heavy} triangles and $\eps$-{\em costly} triangles are bounded by  $2\eps \nt$ and $\eps \nt$ respectively. 
\end{lemma}
\begin{proof}
We begin the proof by first showing that
the number of {costly triangles} is bounded.
To prove this, observe that for a {\em costly edge} $e$,
$t_e < d_e \cdot {(\eps \nt/m\degen)}$.
Then,
$$ \sum_{e \text{ is costly}} t_e < \frac{\eps \nt}{m \degen} \sum_{e \text{ is costly}} d_e < \frac{\eps \nt}{m \degen} \cdot d_E = 2\eps \nt$$,
where the last inequality follows from~\cref{lem:chiba}.

We now turn our attention to bounding the number of {\em 
heavy triangles}. By a simple counting argument,
the number of {\em heavy edges} in $G$ is at most
$\eps \nt / \degen$. Consider the subgraph of $G$ induced
by the set of {\em heavy edges}, denoted as $G_{\heavy}$. 
It follows from the definition of degeneracy that
$\degen_{G_{\heavy}} \leq \degen_G$.
By~\Cref{lem:max_tri}, the number of triangles in $G_{\heavy}$ is then at most $\degen_{G_{\heavy}} \cdot E(G_{\heavy}) = \eps \nt$.
Since any {\em heavy triangle} in $G$ is present
in $G_{\heavy}$, the lemma follows.
\end{proof}

% We are now ready to describe the partial function $\pi$.
% We present it at~\cref{proc:pi} in~\cref{alg:IsAssigned}.
% We give an informal description of the 
% function below.
% Assume $\tau=\{e_1,e_2,e_3\}$ be a triangle. If all three edges are {\em heavy edges}, then 
% $\tau$ is not mapped by $\pi$.
% Otherwise, $\pi$ maps $\tau$ to the edge with 
% $t_e \leq  \degen/ \eps$. If more than one edge satisfies
% the criteria, then chose the lexicographically smallest one.

We give the details of the procedure in~\Cref{alg:IsAssigned}.
The technical part of this method is handled
by \textsc{Assignment} subroutine at~\Cref{line:assgn}.
Given a triangle $\tau$, \textsc{Assignment} either 
returns $\perp$ ($\tau$ is not assigned to any edges)
or returns an edge $e$.
We remark here that the method \textsc{Assignment} as described is randomized and may return different $e$ on different invocations.
To ensure that every triangle is assigned to an unique
edge, as demanded in the item (1) of~\Cref{def:assign},
we maintain a table of (key,value) pairs that maps triangles (key)
to edges or $\perp$ symbol(value). 
In particular, when \textsc{Assignment} is invoked with input $\tau$,
we first look up in the table to check if there is an entry
for the triangle $\tau$. If it is there, then we simply return
the corresponding value from the table. 
Otherwise, we execute \textsc{Assignment} with input $\tau$;
create an entry for $\tau$ and
and store the return value together with $\tau$ in the table.
Since the expected number of calls to the
\textsc{IsAssigned} routine is bounded by $\tO(1)$ (by~\Cref{lem:space}), 
this only adds a constant space overhead.

\begin{algorithm}[!ht]
\caption{Detecting Edge-Triangle Association}
\label{alg:IsAssigned}
    \begin{algorithmic}[1]
    \Procedure{IsAssigned} {triangle $\tau=\{e_1,e_2,e_3\}$, edge $e$}
    \State Let $e_{\min} = $ \textsc{Assignment}$(\tau)$
    \If {$e_{\min} = \perp$ or  $e_{\min}\neq e$} 
    \State return NO.
    \Else 
    \State return YES.
    \EndIf
    \EndProcedure
    \Procedure{Assignment} {triangle $\tau=\{e_1,e_2,e_3\}$}
    \label{line:assgn}
    \For {each edge $e \in \tau$}
    \If {$d_{e} > \frac{m\degen^2}{\eps^2\nt}$} \label{line:costly}
    \State $Y_{e} = \infty$ 
    \Else     
        \For {$j=1$ to $s$}
            \State sample $w$ from $N(e)$ u.a.r.
            \State If $\{e,w\}$ forms a triangle, set $Y_j=1$; \State Else set $Y_j=0$.
        \EndFor
        \State Let $Y_{e} = \frac{d_{e}}{s} \sum_{j=1}^{\ell} Y_j$. 
        \EndIf
    \EndFor
    \State Let $e_{\min} =\argmin_{e} Y_{e}$.
    \If {$Y_{e_{\min}} > \degen / (2\eps)$} 
    \State return $\perp$.
    \Else 
    \State return $e_{\min}$.
    \EndIf
    \EndProcedure
    \end{algorithmic}
\end{algorithm}
We next analyze~\Cref{alg:IsAssigned}. The following theorem 
captures the theoretical guarantees of \textsc{IsAssigned} procedure.
\begin{theorem}
\label{thm:assgn}
Let $\eps>0$ and $c>60$ be some positive constants
and $s=\frac{c\log n}{\eps^2} \cdot \frac{m\degen}{\nt}$.
Then, \Cref{alg:IsAssigned} leads to an $(\eps,o(1/n^5))$-accurate
\textsc{IsAssigned} procedure.
\end{theorem}
In the remaining part of this section, we prove the above theorem. We have already discussed how to ensure \textsc{IsAssigned} satisfies item (1) in~\Cref{def:assign}.
We next take up item (3). We show that not
too many triangles are assigned to any fixed edge.
In particular, we show that if an edge is ``heavy'',
then with high probability no triangles are assigned to it. 
In other words, if an edge $e$ is assigned a triangle
by \textsc{Assigned}, then with high probability $t_e \leq \degen/\eps$. It follows then, that for any edge $e$, the number of triangles that are assigned to $e$, denoted as $\tau_e$, is at most $\degen/\eps$ with high probability.
\begin{lemma}
\label{lem:heavyEdge}
Let $e$ be an $\eps$-heavy edge. Then with probability at least
$1- \frac{1}{n^5}$, no triangles are assigned to $e$.
\end{lemma}
\begin{proof}
First assume $e$ is not an 
$\eps$-costly edge. Let $\tau$ be some triangle that 
$e$ participates in. We consider an execution of 
\textsc{Assignment} on input $\tau$. Clearly, $\Pr[Y_j=1]=t_e/d_e$.
By linearity of expectation,
$\EE[Y_e] = t_e\,.$
An application of Chernoff bound(~\Cref{thm:chernoff})
yields 
\begin{align*}
\Pr[Y_e < \degen / (2\eps) ] &\leq  \Pr[Y_e < t_e/2 ] \\
&\leq \exp\left(-\frac{1}{12}\cdot \frac{sd_e}{t_e} \right) \\
&\leq \exp\left(-\frac{c\log n}{\eps^2}\cdot \frac{st_e}{d_e} \right) \\
& \leq \frac{1}{n^{5}}
\end{align*}
where the last inequality uses the fact that 
$e$ is not $\eps$-costly and hence $d_e/t_e \leq m\degen / (\eps \nt)$.

Next assume $e$ is an $\eps$-costly edge. Since $t_e > \degen/ \eps$, it follows that $d_e > \frac{m\degen^2}{\eps^2 \nt}$.
Then, the if condition on~\Cref{line:costly} is true
and hence $Y_e=\infty$. So no triangles that $e$ participates
in, will be assigned to it.
\end{proof}

We now consider item (2) in~\Cref{def:assign}.
Let $\tau$ be a triangle such that $\tau$ is  
neither $\eps$-heavy nor $\eps$-costly. We prove that,
with high probability, \textsc{Assigned} does not return $\perp$ when invoked with $\tau$. By~\Cref{lem:heavy_costly_triangle},
this implies that $\total \geq (1-3\eps)\nt$.
\begin{lemma}
\label{lem:assigned_triangle}
Let $\tau$ be a triangle that is neither $4\eps$-heavy nor $4\eps$-costly. Then with probability at least $1-o(1/n^5)$, \textsc{Assigned} $(\tau) \neq \perp$.
\end{lemma}
\begin{proof}
Since $\tau$ is not $4\eps$-{\em heavy}, for each edge $e \in \tau$,
$t_e \leq \degen / (4\eps)$. Since $\tau$ is not $4\eps$-{\em costly}, at least one edge is not $4\eps$-{\em costly} --- 
let $e$ denote that edge. Then, $d_e/t_e \leq m\degen / (4\eps \nt)$. Together, they imply $d_e \leq m\degen^2/(16\eps^2 \nt)$, and the 
if condition on~\Cref{line:if} is not met.
We next show that, with high probability $Y_e < \degen / (2\eps)$.

By linearity of expectation, $\EE[Y_e] = t_e$. An application of 
Chernoff bound(~\Cref{thm:chernoff}), similar to the previous lemma,
shows that with probability at least $1-\frac{1}{n^5}$, $Y_e \leq 2t_e \leq \degen / (2\eps)$. Hence, with high probability, the 
triangle $\tau$ is assigned to the edge $e$, proving the lemma.
\end{proof}
We now have all the necessary ingredients to complete the proof
of the theorem.
\begin{proof}[Proof of~\Cref{thm:assgn}]
We have already argued how to ensure that \textsc{IsAssigned}
procedure satisfies item (1) in~\Cref{def:assign}.
\Cref{lem:assigned_triangle,lem:heavyEdge} proves that
\textsc{IsAssigned} satisfies item (2) and item (3) of ~\Cref{def:assign}. Finally, the space bound in item (4) is
enforced by the setting of the parameter $s$.
\end{proof}

\section{Lower Bound}

In this section we prove a multi-pass space lower bound
for the triangle counting problem. Our lower bound, stated below, is effectively optimal.
\begin{theorem} \label{thm:lb1} Any constant pass randomized streaming algorithm for graphs with $m$ edges, $T$ triangles, and degeneracy at most $\kappa$, that provides a constant factor approximation to $T$
with probability at least $2/3$, requires storage $\Omega(m\degen/T)$.
\end{theorem}

% In fact, we prove a more nuanced and arguably more general version of the lower bound. Although, we first lay out some details leading up to it.

Our proof strategy follows along the expected line of 
reduction from a suitable communication complexity problem.
We reduce from the much-studied \textsc{set-disjointness} problem in communication complexity. It is perhaps a 
canonical problem that has been used extensively to prove multi-pass lower bounds for various problems, including triangle counting~\cite{bera2017,Braverman2013}. 
We consider the following promise version of this problem.
Alice and Bob have two $N$-bit binary strings $x$ and $y$ respectively, each with exactly $R$ ones. They want to decide whether there exists an index $i\in [N]$ such that $x_i=1=y_i$. We denote this as the $\disj_{R}^{N}$ problem.

The basis of the reduction is the following lower bound for the
$\disj_{R}^{N}$ problem. Assume $\R(\disj_{R}^{N})$ denote the
randomized communication complexity for the $\disj_{R}^{N}$ problem.
\footnote{See~\cite{kushilevitz1997communication} for the definition of the notion {\em randomized communication complexity}.}
\begin{theorem}[Based on~\cite{kalyanasundaram1992,razborov1992}]
\label{thm:DISJ}
For all $R \le N/2$, we have $\R(\disj_{R}^{N}) = \Omega(R)$.
\end{theorem}

To prove our lower bound, we reduce the $\disj_{R}^{N}$ problem
to the following \textsc{triangle-detection} problem.
Consider two graph families $\mathcal{G}_1$ and $\mathcal{G}_2$;
$\mathcal{G}_1$ is a collection of triangle-free graphs on
$n$ vertices and $m$ edges with degeneracy $\degen$, and 
$\mathcal{G}_2$ consists of graphs on same number 
of vertices and edges, and with degeneracy $\Theta(\degen)$
and has at least $\nt$ many triangles. Given a graph $G\in 
\mathcal{G}_1 \cup \mathcal{G}_2$ as a streaming input, 
the goal of the \textsc{triangle-detection} problem is to decide
whether $G \in \mathcal{G}_1$ or $G \in \mathcal{G}_2$
with probability at least $2/3$, making a constant number
of passes over the input stream. A lower bound for the \textsc{triangle-detection} problem immediately gives a lower bound for the triangle counting problem.

To set the context for our lower bound result, it is helpful 
to compare against prior known lower bounds.
The multi-pass space complexity of the triangle counting
problem is $\Theta(\min\{m^{3/2}/\nt, m/\sqrt{\nt}\})$ \cite{bera2017,McGregor2016,cormode2014} for graphs
with $m$ edges and $\nt$ triangles. 
We first consider the first term $m^{3/2}/{\nt}$.
Since $\degen=O(\sqrt{m})$, our result
subsume the bound of $\Theta(m^{3/2}/\nt)$. 
Compared between $m\degen/\nt$
and $m/\sqrt{T}$, the former is smaller when $\nt > \degen^2$.
Necessarily, in our lower bound proofs, we will be dealing
with graph instances such that $\nt > \degen^2$.

For the purpose of proving a $\Omega(m\degen/{\nt})$ lower bound,
it is sufficient to show that the \textsc{Triangle-Detection} problem
requires $\Omega(m\degen/{\nt})$ bits of space for some
specific choice of parameters. However, we cover
the entire possible range of spectrum. Fix two
parameters $\degen$ and $r$, such that $r\geq 2$.
Then, we can construct an instance of the \textsc{Triangle-Detection} problem with degeneracy $\Theta(\degen)$ and $\nt = \degen^{r}$
such that solving it requires $\Omega(m\degen/{\nt})$ bits of space.
So in effect, we prove a more nuanced and arguably more general theorem than the one give in~\Cref{thm:lb1}. Formally, we show the following.

\begin{theorem}
\label{thm:lower_bound}
Let $\degen$ and $r$ be parameters such that $r\geq 2$.
Then there is a family of instances with {\em degeneracy}
$\Theta(\degen)$ and $\nt = \degen^{r}$ such that solving the 
\textsc{Triangle-Detection} problem requires $\Omega(m\degen/{\nt})$
bits of space.
\end{theorem}

\begin{proof}
We reduce from the $\DISJ^{N}_{N/3}$ problem. Let $(x,y)$ be 
the input instance for this problem. We then construct 
an input $G$ for the \textsc{triangle-detection} problem such that if $(x,y)$ is a YES instance, then $G \in \mathcal{G}_1$, and otherwise $G\in \mathcal{G}_2$. The graph
$G$ has a fixed part and a variable part that depends on $x$ and $y$. We next describe the construction of the graph $G$.

Let $G_{\fixed}=(A\cup B, E_{\fixed})$ be a complete bipartite
graph on the bi-partition $A$ and $B$. Then, $E_{\fixed} = \{\{a,b\}:a\in A,~b\in B\}$. Further assume $|A|=|B|=p$. We add $N$ blocks of vertices to $G_{\fixed}$ and denote them as $V_1,V_2,\ldots,V_N$. Assume $|V_i|=q$ for each $i\in [N]$.
We will set the parameters $p$ and $q$ later in the analysis.
For each index $i\in[N]$ such that $x_i=1$,
Alice connects each every vertex in $V_i$ to each vertex in 
$A$. Denote this edge set as $E_{A}$. For each index $i\in[N]$ such that $y_i=1$, Bob connects each every vertex in $V_i$ to each vertex in $B$. Denote this edge set as $E_B$. 
This completes the construction of the graph $G$.
To summarize, $G=(V,E)$ where $V= A\cup B\cup V_1\cup \ldots V_N$, and $E=E_{\fixed} \cup E_A \cup E_B$.

It is easy to see that the graph $G$ is triangle-free if and only if there does not exits any $i\in [N]$ such that $x_i=1=y_1$. We now analyze various parameters of $G$.
In both YES and NO case for the $\DISJ^{N}_{N/3}$ problem,
we have
\begin{align*}
    n &= |V| = 2p + Nq \,\\
    m &= |E| = p^2 + 2\cdot \frac{N}{3}\cdot pq \,.
\end{align*}
In the NO instance, the number of triangles $\nt$ is at least
$p^2q$. As argued above, in the YES instance, $\nt = 0$.
Finally, we compute the degeneracy $\degen$ in both the cases.
Note that $\degen(G_{\fixed}) = p$, and by definition (see~\Cref{def:degen}) $\degen(G) \geq p$.
We claim that $\degen=p$ in the YES instance and $\degen \leq 2p$
in the NO instance. To prove the claim, we use the following
characterization of degeneracy. Let $\prec$ be
a total ordering of the vertices and let $d^{\prec}_{v}$ denote the 
number of neighbors of $v$ that appears after $v$ according to the ordering
$\prec$. Let $d_{\max}^{\prec} = \max_{v\in V} d^{\prec}_{v}$. Then,
$\degen \leq d_{\max}^{\prec}$. 
Now consider the following ordering: $V_1 \prec V_2 \prec \ldots V_N \prec A \prec B$, and inside each set the vertices are ordered arbitrarily. Then, in the YES instance, $d_{\max}^{\prec} \leq p$
and in the NO instance $d_{\max}^{\prec} \leq 2p$, proving our claim.

We now set the parameters $p$ and $q$ as $p=\degen$ and $q=\degen^{r-2}$. Then, $m=\Theta(Npq)$ since $p=O(Nq)$. Assume there is a constant pass $o(m\degen/\nt)$-space streaming algorithm $\mathcal{A}$ for the \textsc{Triagnle-Detection} problem. Then, following standard reduction, $\mathcal{A}$ can be used to solve the $\DISJ_{N/3}^{N}$ problem with $o(Npq \cdot p / p^2q) = o(N)$ bits of communication, contradicting the lower bound for the $\DISJ_{N/3}^{N}$ problem. 
\end{proof}

\section{Future Directions}

In this paper, we studied the streaming complexity of the 
triangle counting problem in bounded degeneracy graphs. 
As we emphasized in the introduction, low degeneracy is an
often observed characteristics of real-world graphs. 
Designing streaming algorithms with better bounds 
on such graphs (compared to the worst case) is an important 
research direction. 
There have been some recent successes in this context
--- graph coloring~\cite{bera2019graph},
 matching size estimation~\cite{assadi2017estimating,esfandiari2018streaming,cormode2017sparse}, independent
set size approximation~\cite{cormode2017independent}.
It would be interesting to explore what other problems can admit better streaming algorithms in bounded degeneracy graphs. One natural candidate is 
the arbitrary fixed size subgraph counting problem, which asks for the
number of occurrences of the subgraph in the given input graph.

\medskip

\emph{Do there exist streaming algorithms for approximate subgraph counting on low degeneracy graphs that can beat known worst-case lower bounds?}

We conclude this exposition with the following conjecture about 
the fixed size clique-counting problem.

\begin{conjecture}
\label{conj:clique}
Consider a graph $G$ with degeneracy $\degen$
that has $\nt$ many $\ell$-cliques. 
There exists a constant pass streaming algorithm that outputs
a $(1\pm \eps)$-approximation to $\nt$ using $\tO(m\degen^{\ell-2}/\nt)$
bits of space.
\end{conjecture}

\begin{acks}
The authors would like to thank the anonymous reviewers for
their valuable feedback. The authors are supported by NSF TRIPODS grant CCF-1740850, NSF CCF-1813165, CCF-1909790, and ARO Award W911NF1910294.
\end{acks}

\bibliographystyle{ACM-Reference-Format}
\bibliography{source/refs,source/sesh}

% \appendix

\end{document}